\documentclass[12pt]{article}

\usepackage{amssymb}
\usepackage{amsfonts}
\usepackage{amsmath}
\usepackage{txfonts}
\usepackage{hyperref}
\usepackage{lineno}
\usepackage{graphicx}
\graphicspath{{converted_graphics/}}

\newtheorem{theorem}{Theorem}

\newtheorem{lemma}[theorem]{Lemma}

\def\R{{\cal{R}}}
\newcommand{\be}{\begin{equation}}
\newcommand{\ee}{\end{equation}}

\newenvironment{proof}[1][Proof]{\noindent\textbf{#1.} }{\ \rule{0.5em}{0.5em}}

\def\R{{\cal{R}}}

\begin{document}

\vskip -30pt

\title{How much is your Strangle worth?\\ {\large \it{On the relative value of the  $\delta-$Symmetric Strangle under the Black-Scholes  model} }}

\author{Ben Boukai \footnote{bboukai@iupui.edu; Tel: +13172746926; Fax: +13172743460.}\\
Department of Mathematical Sciences, IUPUI, Indianapolis,\\
IN 46202 , USA}

\maketitle

\begin{abstract}
Trading option strangles is a highly popular strategy often used by market participants to mitigate volatility risks in their portfolios. In this paper we propose a measure of the relative value of a delta-Symmetric Strangle and compute it under the standard  Black-Scholes option pricing model. This new measure accounts for the price of the strangle, relative to the Present Value of the spread between the two strikes, all expressed, after a natural re-parameterization,  in terms of delta and a volatility parameter.  We show that under the standard BS option pricing model, this measure of relative value is bounded by a simple function of  delta only and is independent of the time to expiry, the  price of the underlying security or the prevailing volatility used in the pricing model.  We demonstrate how this bound can be used as a quick {\it benchmark} to assess, regardless the market volatility, the duration of the contract  or the price of the underlying security,  the market (relative) value of the  $\delta-$strangle  in comparison to its BS (relative) price. In fact, the explicit and simple expression for this measure and bound allows us to also study in detail the strangle's exit strategy and the corresponding {\it optimal} choice for a value of delta.    


\end{abstract}

\section{Introduction}

Options, as asset's price derivatives, are the primary tools available to the market participants for hedging their portfolio from directional risk and/or volatility risk. 
The so-called option's delta, which typically is denoted as $\delta$ or $\Delta$,  measures the 'sensitivity' of the option's price to changes in the  price of the  underlying security, is the primary parameter one considers when using an option to mitigate directional risk.  The option's delta is  seen as the hedging ratio and is often also used (near expiration) by market participants as a surrogate to the probability  that the option will expire in the money. 
With standard option pricing model of  Black and Scholes (1973), (abbreviated here as the BS model,  see below), these probabilities are readily available for direct calculations under the governing log-normality assumption of the asset's returns. Roughly speaking, a trader that sells (or buys) a put option at a strike located one standard deviation {\it below} the  current asset's price, ends up with a 16-delta put contract option (i.e. with $\delta=-0.16$). We denote the corresponding strike 
for this 16-delta put contract option as $k_{0.16}^-$. Similarly,   a trader that sells (or buys) a call option at a strike located one standard deviation {\it above} the  current asset's price, ends up with a 16-delta call option (i.e. $\delta=0.16$). We denote the corresponding strike for this 16-delta call contract as $k_{0.16}^+$.  Thus, the corresponding strangle, which is obtained by selling a (negative) 16-delta put option and a (positive) 16-delta call option, is a delta-neutral strategy that is associated, very roughly, with a 0.68  probability for the asset's price to remain between the two strikes, $k_{0.16}^-$ and $k_{0.16}^+$ by expiration, all as resulting from the governing normal distribution assumption. We refer to such a strangle as a {\it   16-delta Symmetric Strangle}, only to indicate the common (absolute) delta value ($\delta=0.16$) of its put and call components.

In a similar fashion we use the term a {\it $\delta-$Symmetric Strangle} to indicate the strangle obtained, for some fixed $\delta\in(0,0.5)$, from  buying (or selling) a $\delta$-units put and call option contracts at the corresponding strikes $k_{\delta}^-$ and $k_{\delta}^+$, respectively. Such a  strangle would be a delta-neutral strategy offering zero directional risk but potentially useful for mitigating volatility risk. We further denote by $\Pi_\delta$ the price of (or the credit received from)  such 
{\it $\delta-$Symmetric Strangle}. In this paper, we study, for a a given $\delta$, the value of this {\it $\delta-$Symmetric Strangle} {relative} to the width of the corresponding spread $(k_{\delta}^+-k_{\delta}^-)$, adjusted for its present value (PV). More precisely, for any  a $\delta \in (0, 0.5)$,  we  define the {\bf relative value} of the corresponding {\it $\delta-$Symmetric Strangle} as
\be\label{1.0}
R_\delta:= \frac{ \Pi_\delta}{PV(k_{\delta}^+-k_{\delta}^-)}.  
\ee

In Section 2, we show  that under the standard BS option pricing model,  the strangle's relative value,  $R_\delta$, is independent of the price of the underlying security and is a function only of  $\delta$ and the prevailing volatility used in the pricing model. In fact, as we will see in Theorem 1 below, for any given $\delta \in (0, 0.5)$, we have $R_\delta\leq \R_\delta$  where
\be\label{1} 
\R_\delta= -\frac{\phi(z_\delta)}{z_\delta} -\delta, 
\ee
and where $\phi(\cdot)$ is the standard normal density ($pdf$), $\phi(u):= \frac{1}{\sqrt{2\pi}}e^{-\frac{u^2}{2}}$, and $z_\delta\equiv \Phi^{-1}(\delta)$, is usual  $\delta^{th}$ percentile of the standard normal distribution, whose cumulative distribution function ($cdf$)  is $\Phi(z):=\int_{-\infty}^z\phi(u)du$.  We point out that since $\delta<0.5$, we have $z_\delta<0$ in expression (\ref{1}) of  $\R_\delta$. 

As an illustration, one quickly finds by utilizing (\ref{1}) that the {\it  $16$-delta Symmetric Strangle} has a relative value of $\R_{0.16}=0.08467$ and that the {\it $30$-delta Symmetric  Strangle} has a relative value of $\R_{0.30}=0.36$. That is to say that under the standard BS option pricing model, one would expect the price of the  {\it  $30$-delta Symmetric Strangle} to be {\it at most} 36\% of the width of the spread between the corresponding strikes, irrespective of the security's  price, or time to expiry, and irrespective of the prevailing volatility. More generally, it follows from Theorem 1, that for a any given $\delta \in (0,\, 0.5)$, the corresponding $\delta-$strangle's price, $\Pi_\delta$, {\it  as calculated under the BS pricing model}, satisfies  
$$
\Pi_{\delta}\leq \R_\delta\times PV(k_{\delta}^+-k_{\delta}^-), 
$$
irrespective of the security's  price, or time to expiry, and irrespective of the prevailing volatility.  In Section 3, we illustrate how this measure $\R_\delta$ in (\ref{1}) may be used as a {\it benchmark} to assess the market pricing (or 'worthiness') of the $\delta-$symmetric strangle compared to its (relative) price,  $R_\delta$,  as suggested by standard BS pricing model. The explicit expression of $\R_\delta$ as is given in (\ref{1})  allow us to also address, in Section 4, the strangle's exit strategy and the corresponding {\it optimal} choice of $\delta$ for it.

\section{Pricing the $\delta$-unit option contract}

One of the most widely celebrated option pricing model for equities (and beyond) is that of Black and Scholes (1973).  Their pricing model is derived under some simple assumptions concerning the distribution of the asset's returns, coupled with presumptive continuous hedging, zero dividend, risk-free interest rate, $r$,  and no cost of carry or transactions fees.  While the aptness of these assumptions has often been criticized (see for example Yalincak (2012)), it has remained as a leading option pricing model for the retail trading practitioner (e.g.: Sinclair (2010)).  However, in its standard form, the BS model  evaluates, for a risky asset with a current market price $\mu$, the price   of an European call option contract at a strike $k$ and $t$ days to expiration as:
\be\label{2}
c_\mu(k)=\mu\times \Phi(d_1(k))-k\cdot e^{-rt}\times\Phi(d_2(k)).
\ee
Here, using the standard notation, 
\be\label{3}
d_1(k):=\frac{\log(\frac{\mu}{k})+(r+\frac{\sigma^2}{2})t}{\sigma\sqrt{t}} \qquad \textnormal{and} \qquad d_2(k):=d_1(k)-\sigma\sqrt{t},  
\ee
 where $\sigma$ denotes the standard deviation of the {{daily}} asset's returns, and $\Phi(\cdot)$ is the standard normal cdf defined above.   The  model for the corresponding price of a put option contract, $p_\mu(k)$,  may be obtain from expression (\ref{2}) of $c_\mu(k)$, by exploiting the so-called {\it put-call parity} which is expressed by the equation
\be\label{4}
\mu-c_\mu(k)=k\cdot e^{-rt} -p_\mu(k),  
\ee
see for example Jiang (2005, Theorem 2.3) or Peskir and Shiryaev (2002) for details. This parity implies that the price of the corresponding put option contract is, 
\be\label{5}
p_\mu(k)=k\cdot e^{-rt}\times\Phi(-d_2(k))-\mu\times [1-\Phi(d_1(k))].
\ee

There is substantial body of literature dealing with the BS option pricing model in (\ref{2})-(\ref{5}), its refinements, its extensions and the so-called, its implied 'Greeks' (i.e. the various partial derivatives of different orders, representing the model's "sensitivities" to changes in its parameters). The interested reader is referred to standard textbooks such as Wilmott, Howison, Dewynne (1995), Hull (2005), Jiang (2005) or Iacus (2011). 

As we already mentioned in the Introduction,  we focus  our  attention here on the option's {\it delta}, which we denote by $\Delta$ as a function with a corresponding value of $\delta\in (0,1)$.  More specifically, while suppressing (for sake of simplicity for now)  from the 
notation $r, t$ and $\sigma^2$, we define for the call and the put contracts options their respective $\Delta$ functions as, $\Delta_c(k):= {\partial c_\mu(k)}/{\partial \mu}$ and $\Delta_p(k):= {\partial p_\mu(k)}/{\partial \mu}$. 
It follows immediately from the put-call parity equation in (\ref{4}) that $\Delta_p(k)=-(1-\Delta_c(k))$. It is well known (see for Example, Jiang (2005)) that for the BS pricing model in (\ref{2}), 
$\Delta_c(k)= \Phi(d_1(k))$, where $d_1(k)$ is given in (\ref{3}), and hence $\Delta_p(k)= -(1-\Phi(d_1(k))\equiv -\Phi(-d_1(k))$.

For its supreme importance to portfolio hedging, the investor/trader often needs to buy (or sell)  an option at a strike, $k$,  which is associated with a {\it specified and desired value} $\delta$ of the option's $\Delta$. For any given $\delta\in (0,1)$, we let  $k^+_\delta$ denote the (unique) solution of the equation $\Delta_c(k^+_\delta)=\delta$, or equivalently the solution of 
\be\label{8}
\Phi(d_1(k^+_\delta))=\delta.
\ee
Accordingly, it follows immediately  that $k^+_\delta$ satisfies  the equation 
\be\label{9}
d_1(k^+_\delta)=\Phi^{-1}(\delta)\equiv z_\delta,
\ee
and hence, by utilizing (\ref{3}) in (\ref{9})  leads to the solution as 
\be\label{10}
k^+_\delta=\mu\cdot e^{-z_\delta \nu + \nu^2/2+r t},
\ee
 where we have substituted  $\nu \equiv \sigma\sqrt{ t}$ throughout. It should be clear from (\ref{10}) that if  $\delta<0.5$, one has $z_\delta<0$ and therefore  $k^+_\delta>\mu$, so that  the corresponding call option is said to be 'out of the money' (OTM). Also, note that it follows from (\ref{3}) and (\ref{9}) that $d_2(k_\delta^+)=d_1(k^+_\delta)-\sigma\sqrt{t}\equiv z_\delta-\nu$, so that 
\be\label{11}
\Phi(d_2(k_\delta^+))= \Phi(z_\delta-\nu)
\ee
in (\ref{2}).  Indeed, with the re-parameterization by $(\delta, \, \nu)$ (with $\nu\equiv \sigma \sqrt{t}$), of the BS option pricing model in (\ref{2}), we may re-express, upon using  the matching expressions (\ref{8})-(\ref{11}) in equation (\ref{2}),  the {current} price of a {$\delta-$unit} call option  in a much simpler form  as
\be
\begin{aligned}\label{12}
c_\mu(\delta, \, \nu)\equiv  c_\mu(k^+_\delta)=&  \mu \times \delta -k^+_\delta \cdot e^{-rt}\times \Phi(z_\delta-\nu)\\ 
= &  \mu \times \left[\delta -e^{-z_\delta \nu +\nu^2/2} \times \Phi(z_\delta-\nu)\right],\\ 
\end{aligned}
\ee
for any $\delta\in (0,1)$ and with $\nu>0$. 
\smallskip

\noindent{\bf Remark 1:\ } {\it 
Note in passing that  in practice, the option's $\delta$ is often used as a crude approximation to the probability the option will end in the money,  $Pr(ITM)$, which by (\ref{11}), (\ref{12}) is equal to $\Phi(d_2(k^+_\delta))\equiv \Phi(z_\delta-\nu)$. However, since $\nu\equiv \sigma\sqrt{t}>0$, it immediately follows that $\Phi(z_\delta -\nu))\leq \Phi(z_\delta))\equiv \delta$. Hence, for any $\delta\in (0,1)$ and $\nu>0$, $Pr(ITM)\leq \delta$ and only near expiration, as $lim_{t\to 0}Pr(ITM)=\delta$, it holds. 
}

Similarly to (\ref{12}), we calculate the {current} price of the $\delta-$unit put contract option by using the put-call parity equation in (\ref{4}), and by noting that by (\ref{8}) the corresponding $k^-_\delta$ strike for the put contract is the same as the strike $k^+_{1-\delta}$ of the $(1-\delta)-$unit call option contract, so that $k^-_\delta\equiv k^+_{1-\delta}$.  Accordingly,  since $z_\delta\equiv -z_{1-\delta}$,  we obtain from (\ref{10}) that 
\be\label{13}
k^-_\delta=\mu\cdot e^{z_\delta \nu + \nu^2/2+r t}.
\ee
Hence, it follows immediately from (\ref{4}) and (\ref{13}), that under the $(\delta, \, \nu)$ re-parameterization,   the expression for the {current} price of the $\delta-$unit put option is,  
\be
\begin{aligned}\label{14}
p_\mu(\delta, \, \nu)\equiv p_\mu(k^-_\delta)= & -\delta \mu +k^-_{\delta}\cdot e^{-rt}\times(1- \Phi(z_{1-\delta}-\nu))\\
 = &  -\mu\times\left[\delta -e^{z_\delta \nu +\nu^2/2} \times \Phi(z_\delta+\nu)\right].\\
\end{aligned}
\ee

{It should be clear from (\ref{13}) that if  $\delta<0.5$ and $r=0$, one has $z_\delta<0$ and therefore  $k^-_\delta<\mu$ only if $\nu< -2\cdot z_\delta$, in which case, the corresponding put option is said to be 'out of the money' (OTM). Hence, we will restrict our attention to the practical case of the above parametrization in which $(\delta, \nu)$ are such that $k^-_{\delta}<\mu < k^+_\delta$, or alternatively, $(\delta, \nu)\in {\cal B}$, where 
$$
{\cal B}=\{(\delta, \nu); \, \, \delta>0, \ \& \ \nu>0,  \ \ \text{s.t.}\ \ \delta<\Phi({-\nu}/{2}) \}.
$$
We further point out that the two strikes, $k^+_{\delta}$ and $k^-_\delta, (\equiv k^+_{1-\delta})$,  need not be symmetrical with respect of the current price $\mu$ of the underlying security (i.e.: $\mu-k^-_{\delta}\neq k^+_{\delta}-\mu$). It is well-known that the occasional observed asymmetry of these equal $\delta-$ units strikes is a fixture of the {\it skew} in the volatility surface that is affecting the option pricing model, see for example Gatheral (2006), or Doran and Krieger (2010).  
}

\section{The relative value of the $\delta$-Symmetric Strangle}

Consider now a trader that desires to simultaneously sells (say), at some given level of $\delta<0.5$, the $\delta-$unit put and the $\delta-$unit call contracts so as to form the 'OTM' {\it $\delta-$Symmetric Strangle} strategy. The total  selling price of this 
{strangle} as calculated under the BS pricing model,  is therefore  $\Pi_\delta=c_\mu(k^+_\delta)+p_\mu(k^-_\delta)$. As a measure for assessing the 'worthiness' of this strangle, we consider the 'value' of the selling price, $\Pi_\delta$, relative to the present value of the spread between the strikes, namely, $PV(k^+_{\delta}-k^-_\delta)=(k^+_{\delta}-k^-_\delta)\times e^{-rt}$. We express this relative value measure in (\ref{1.0}) as

\begin{equation}\label{15.001}
R(\delta, \, \nu):= \frac{\Pi_\delta}{PV(k^+_{\delta}-k^-_\delta)} =  \frac{c_\mu(\delta, \nu)+p_\mu(\delta, \nu)}{(k^+_{\delta}-k^-_\delta)\times e^{-rt}}.
\end{equation}
\smallskip

Note that by its definition, $R(\delta, \, \nu)\geq 0$ for all $\delta\in (0,0.5)$ and $\nu>0$,  in particular over ${\cal B}$.  {Further,  since an European option price is (linearly) homogeneous in  $\mu$,  and in the strike, $k$, (see Theorem 6 of Merton (1973)),  the ratio $R(\delta, \nu)$ in (\ref{15.001}),  is independent of the current price,   $\mu$, of the underlying security. Also note that since we account in  (\ref{15.001}) for the present value of the spread between the strikes, this quotient is, by construction, also independent of  the risk-free interest rate,  $r$.}
This can fully realized by substituting  expressions (\ref{10}), (\ref{12}), (\ref{13}) and (\ref{14}) in $R(\delta, \, \nu)$ and simplifying the resulting terms, to obtain, for each $\delta<0.5$ and $\nu\equiv \sigma\sqrt{t}>0$,  the expression, 
\be\label{15}
R(\delta, \, \nu)=  \frac{e^{z_\delta \nu}\cdot \Phi(z_\delta+\nu)- e^{-z_\delta \nu }\cdot \Phi(z_\delta-\nu)}{e^{-z_\delta \nu }-e^{z_\delta \nu }} ,
\ee
for the  relative value of the {\it $\delta-$Symmetric Strangle}  {under the BS option pricing model.}  {We point that the values of $R(\delta, \, \nu)$ in  (\ref{15})  are straightforward to calculate for any $(\delta, \nu)$}. Figure 1 below provides the graph of $R(\delta, \, \nu)$ for various values of $(\delta, \nu)$, with $0<\delta<0.5$ and $0<\nu<1$, where $\nu=\sigma\sqrt{t}$  representing realistic values for $t$ (the time in days to expiry) and the model's daily (implied or historical) volatility, $\sigma$.  In any case, the properties of $R(\delta, \, \nu)$, as a function of $\delta$ and $\nu$ (in ${\cal B}$)  are of interest. In Appendix A below we show that for a fixed $\delta<0.5$, $R(\delta, \nu)$ is monotonically non-increasing function of $\nu$ (with $\partial R/\partial\nu \leq 0$) and that for a fixed $\nu>0$, $R(\delta, \nu)$ is monotonically increasing function of $\delta$ (with  $\partial R/\partial\delta > 0$).  

\begin{figure}[h] 
  \centering
  \includegraphics[width=4in,height=3in,keepaspectratio]{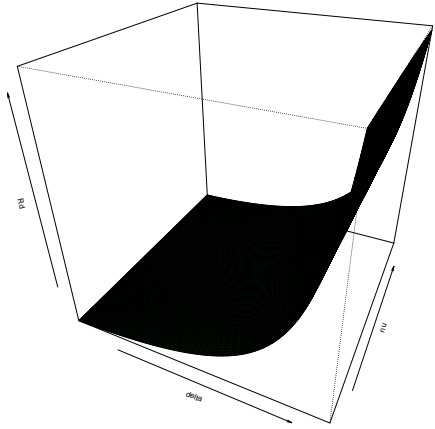}
  \caption{\small{\it The Relative Value function $R(\delta, \nu)$ of the {$\delta-$Symmetric Strangle} for $\delta<0.5$ and $\nu\in (0,1)$.}}
  \label{fig:fig1}
\end{figure}

\begin{theorem} Under the BS model and irrespective of the current price, $\mu$,  of the underlying security,  the current risk-free interest rate, $r$, and irrespective of the time to expiry, $t$, and the presumed volatility (either implied or historical), the upper bound to the relative value $R(\delta, \nu)$, of the OTM {\it $\delta-$Symmetric Strangle} with $\delta\in (0, 0.5)$, depends only on $\delta$ and is given by, 
$0<R(\delta, \nu)\leq  \R_\delta$, where  
\be\label{15.1}
\R_\delta:=\lim_{\nu\to 0} R(\delta, \nu)=-\frac{\phi(z_\delta)}{z_\delta} -\delta. 
\ee
Moreover,  $\lim_{\delta\to 0}\R_\delta=0$, and for all $\delta\in (0,0.5)$,
\be\label{15.2}
\R^\prime_\delta:=\frac{d}{d\delta}\R_\delta=\frac{1}{z^2_\delta}>0. 
\ee

\end{theorem}

\begin{proof} That $R(\delta, \nu)$ is a monotonically decreasing function of $\nu$ for each fixed $\delta \in (0,0.5)$ is seen by direct calculation, $\partial R(\delta, \nu)/\partial \nu\leq 0$ (see Lemma \ref{lm3}, below). The results stated in (\ref{15.1}) follow immediately by a straightforward application of L'Hopital's rule to the numerator and denominator that comprise expression ({\ref{15}) of  $R(\delta, \nu)$ and noting that it trivially also independent of $\mu$ and $r$ by construction. By another direct application of L'Hopital's rule to the quotient 
$\phi(z_\delta)/z_\delta$ along with the facts that $\frac{d}{d\delta}\phi(z_\delta)=-z_\delta\phi(z_\delta)z^\prime_\delta$ and $z^\prime_\delta=1/\phi(z_\delta)$ leads to the second assertion as well as to the result stated in (\ref{15.2}). }\end{proof}

The results of Theorem 1 and the bound  $\R_\delta$ in (\ref{15.1}) provide  a {\it benchmark} for assessing the value, in relative terms, of a {\it $\delta-$Symmetric Strangle} under the BS option pricing model in (\ref{2})-(\ref{5}), as applicable to any security (i.e. independent of the current underlying security price $\mu$), to any expiry (independent of $t$), and under any presumed volatility (independent of $\sigma$).  In fact, if $\hat R_\delta$ denotes the market (relative) value of a {\it $\delta-$Symmetric Strangle} (i.e. the market version of (\ref{1.0})), then, this strangle would be deemed  {\it 'well-priced'}  compared to its (relative) price under the BS option pricing model, as long as $\hat R_\delta\geq \R_\delta$.  In Figure 2 below, we graph the values of this function,   $\R_\delta$ (in (\ref{15.1}) or (\ref{1})) for all $0<\delta <0.5$.

\begin{figure}[h] 
  \centering
  \includegraphics[width=3.5in,height=2.5in,keepaspectratio]{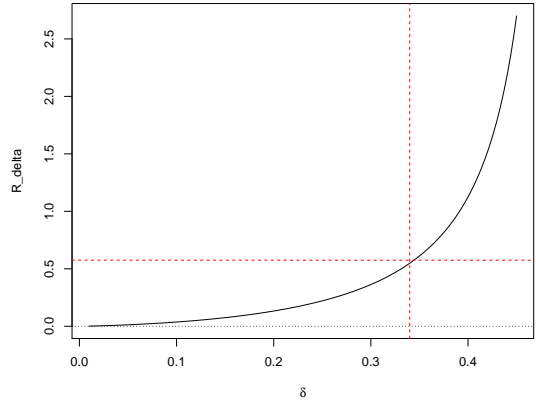}
  \caption{\small{\it The relative value $\R_\delta$ as a function of $\delta$. Marked in {red} is the market relative value (current, as of EOD, May, 13th, 2020), $\hat R_{0.34}=0.575$,  of a 34-delta symmetric strangle with strikes \$112 and \$120 in IBM (see Example 1 and Table 1 for more details.)}}
  \label{fig:fig2}
\end{figure}

\noindent{\bf Remark 2:\ }{\it The results stated in Theorem 1 and their derivations are valid in the BS 'world', in which the distribution of the asset's returns assumed to have a constant variability throughout and do not take into account the volatility 'skew' or 'smile' that is often being observed by the traders across the discretized options' grid equipped with bid-ask price spreads.  {It surely implies that the BS pricing model (with all it inputs) undervalues the $\delta$-Symmetric Strangle, whenever $\R_\delta <\hat R_\delta$, where  $\hat R_\delta$ is its market (relative) value (i.e. the market version of (\ref{1.0})). }}

\vskip 5pt

\noindent{\bf Example 1:\ }{ 
As an illustration of it's usage, consider the market EOD (end of day) market pricing of IBM (International Business Machine Corp.) as of May 13th, 2020. We find  that the $34-$delta symmetric strangle for the June 5th, 2020 expiration with the strikes of $k_1=\$112$ and $k_2=\$120$ for the sold put and call,  respectively, has a market mid-price of $\hat \Pi_{0.34}=\$4.60$ (along with current ticker price of $\mu=$\$115.73, with $t=23$ days to expiration, and $IV=38.32\%$ (average) implied volatility, so that $\sigma=IV/\sqrt{365}=0.0200576$).   This results with a market relative value of 
$$
\hat R_{0.34}:= \frac{\hat \Pi_{0.34}}{(k_2-k_1)}= \frac{4.60}{(120-112)}=0.575, 
$$
for this 34-delta  Strangle in IBM, whereas, by using (\ref{1}), we calculate under the BS pricing model a relative value of $\R_{0.34}=0.548$ for this 34-delta strangle. Thus, the BS pricing model (with its constant variance assumption, etc.)  under-values this strangle (in relative terms) as compared to its actual market value. Similar result is obtained with the relative value of a 21-delta strangle with 100 days to expiration in BA (Boeing Co.), which yields $\hat R_{0.21}=0.156$ as compared to $\R_{0.21}=0.147$, see Table 1 below.  Also included in this table the the market pricing of $\delta-$ symmetric strangles for additional securities,  with different $\delta$, underlying prices, $IV$ and days to expiration.   {In all cases listed in the Table, the market (relative) value $\hat R_\delta$ exceeded that of the corresponding BS (relative) value $\R_\delta$. Thus, in these noted cases, the BS pricing model (with all its inputs) appears to undervalue the strangles (in relative terms) as compared to their market (relative) value.} The reader is invited to check the validity of Theorem 1  results and the applicability of the bound $\R_\delta$ in (\ref{15.1})  as a benchmark for the market pricing (in relative value) of a  $\delta$-Symmetric Strangle with any other traded security options at any expiration.

\begin{table}[h]
\begin{center}
\caption{Oserved market relative values $\hat R_\delta$ of the $\delta-$Symmetric Strangle for various tickers and durations as were priced on EOD$^*$, May 13, 2020,  as compared to the bound  $\R_\delta$ (\ref{15.1}) calculated under the BS option pricing model.}
{\small
\begin{tabular}{ccccccccccc}
\hline
& $Ticker$ &  $\mu$ 	&  $IV$	& $Days$ & $\delta$ & $k_1$  & $k_2$ & $\hat \Pi_\delta$-Price & $\hat R_\delta$ & $\R_\delta$\  \\ \hline
& SPY	 & 281.60 	&0.3529	&	37 &	0.170	& 250 & 302	 &	5.19	&	0.107	& 0.095 \  \\ 
 & LLY	 & 157.93	& 0.3597&	156 &	0.200	& 130 & 185	 &	8.22 	&	0.150	& 0.133 \  \\
& BA	 & 121.50 	& 0.7685&	100 &	0.210	& 95 & 175	 &	12.45 	&	0.156	& 0.147 \  \\ 
& TLT	 & 168.50	& 0.2029&	16 &	0.255	& 162 & 170	 &	1.97 	&	0.246	& 0.232\  \\ 
& C	 & 40.60	& 0.6851&	219 &	0.295	& 32& 52.5	 &	7.05 	&	0.362	& 0.345 \  \\ 
& IBM	 & 115.73  	&0.3832	&	23 &	0.340	& 112 & 120	 &	4.60	&	0.575	& 0.548 \  \\ 
& GOOG	 & 1349.33	& 0.3356&	65 &	0.405	& 1320& 1400	 &	102.65 	&	1.283	& 1.207 \  \\ 
\hline
\end{tabular}
}
\end{center}
\vskip -10pt
\qquad \small{$^*$EOD market pricing were obtained using the TOS platform of TDAmeritrade} 
\end{table}

 {The empirical (market) results exhibited in Table 1, illustrate the tremendous practical and strategy implications the results of Theorem 1 have. With such  theoretical results at hand, retail traders and and market participants  are now able to quickly assess whether or not the  strangle they buy (or sell) is overpriced or under-priced in the market as compared to its Black-Scholes price, regardless the market volatility, the duration of the contract  or the underlying assest's price . It provides for a common \textit{benchmark} for assessing and comparing the market (relative) pricing of a major trading strategy (namely the $\delta-$Symmetric Strangle) across various securities and assets, across various durations and irrespective of the underlying security-specific volatility (implied or historical).}

\section{Strategizing}
   
One of the appealing aspects of a {\it $\delta-$Symmetric Strangle} is that from the outset, it is a  delta-neutral strategy with zero directional risk, initially. Moreover a trader that sells such a strangle, for some fixed $\delta<0.5$, at the matching two strikes $k^-_\delta$ and $k^+_\delta$, benefit from a well defined probability of success, that may be calculated under the {\it current}  distribution of the asset's returns  implied by BS option pricing model in (\ref{2}) and (\ref{5}).  Specifically, for a given value $\delta<0.5$ and $\nu>0$, the {\it initial } probability that the underlying security price would remain, at expiration, between $k^-_{\delta}, \ (\equiv k^+_{1-\delta})$ and $k^+_{\delta}$ is simply (see Remark 1),  
\be\label{16}
\alpha \equiv \Phi(-z_{\delta}-\nu)-\Phi(z_{\delta}-\nu). 
\ee
Hence, the expected reward for a trader that sells the strangle for $\Pi_\delta=c_\mu(k^+_\delta)+p_\mu(k^-_\delta)$ (as credit) and plans to exit and buy it  back for a fraction $\lambda\in (0,1]$ of the credit received is
$$
E_\lambda(\delta):=\alpha\Pi_\delta-(1-\alpha)\lambda \Pi_\delta. 
$$
In relative terms, this expected reward, relative to the present value of the spread between the strikes, becomes 
\be\label{17}
{\bar E}_\lambda(\delta):= \frac{E_\lambda(\delta)}{PV(k^+_{\delta}-k^-_{\delta})}\leq \alpha\R_\delta-(1-\alpha)\lambda \R_\delta, 
\ee
where $\R_\delta$ is given in (\ref{1}). As was mentioned in the Introduction and pointed out in Remark 1, for small values of $\nu$ (i.e. near expiration) we may approximate the {\it 'success'} probability in (\ref{16}) as
$\alpha\approx (1-2\delta)$. Accordingly, for any given fractional loss $\lambda\in(0,1]$ the expected relative reward in (\ref{17}), under this approximation would be,
\be\label{18}
{\cal E}_\lambda(\delta)=(1-2\delta(1+\lambda))\times\R_\delta. 
\ee

\begin{figure}[h] 
  \centering
 \includegraphics[width=3.75in,height=3in,keepaspectratio]{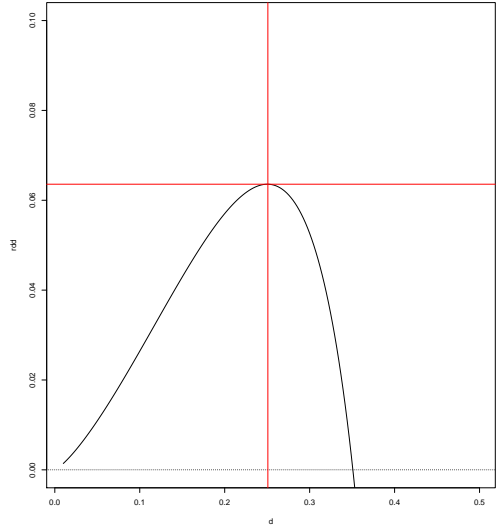}
  \caption{\small {\it The expected relative reward function, ${\cal E}_\lambda(\delta)$ as a function of $\delta$ with $\lambda=0.5$. The maximal value is achieved at $\delta^*=0.2336$ at which point, ${\cal E}_\lambda(\delta^*)=0.05615$.  }}
  \label{fig:fig3}
\end{figure}

Observe that ${\cal E}_\lambda(\delta)\geq 0$ as long as $\delta\leq 1/2(1+\lambda)$ and that,  upon using (\ref{15.1}) and (\ref{15.2}), the equation 
$$
{\cal E}^\prime_\lambda(\delta)\equiv -2(1+\lambda)\R_\delta+(1-2\delta(1+\lambda))\R^\prime_\delta=0,
$$
is seen to have a unique root,  $\delta^*$,  at which point ${\cal E}_\lambda(\delta)$ attains its maximal value.  That is, for a given fractional lose $\lambda\in (0, 1]$, this root $\delta^*=h(\lambda)$, must solves the equation 
\be\label{19}
\delta(1-z_\delta^2)-z_\delta\phi(z_\delta)=\frac{1}{2(1+\lambda)},
\ee
at which point ${\cal E}^*_\lambda:={\cal E}_\lambda(\delta^*)\geq {\cal E}_\lambda(\delta)$.

In Figure 3 above we present the graph of  the relative reward function,  ${\cal E}_\lambda(\delta)$,  for a trader who sells a $\delta-$Symmetric Strangle, and wishes, as a matter of strategy,  to exit it upon a loss of 50\% of the credit received. This case corresponds to $\lambda=0.5$ and results with an optimal choice for $\delta$ of   $\delta^*=0.2336$ for this strategy to yield a maximal expected relative reward of ${\cal E}^*_{0.5}=0.05615$.

\begin{table}[h]
\begin{center}{\small 
\caption{The optimal choice for $\delta$ for the 
 $\delta-$ Symmetric Strangle strategy, calculated  for  'exits'  with the various fractional loss $\lambda$.}

\vskip 8pt

\begin{tabular}{ccccc}
\hline
&  $\lambda$ & $\delta^*$ & ${\cal E}^*$ & $\alpha(\delta^*)$ \  \\ \hline
& 0.25	 &	0.300 &	0.091	& 0.400 \  \\ 
& 0.40	  &	0.256 &	0.067	& 0.489 \  \\ 
& 0.50	&	0.234&	0.056	& 0.533	\  \\ 
& 0.60	&        0.216 &	0.048	& 0.567	\  \\ 
& 0.75	&	0.194&	0.040	& 0.611\  \\ 
& 1.00	&	0.164&	0.031	& 0.670\  \\ \hline

\end{tabular}
}
\end{center}
\end{table}

In Table 2 below, we provide the 'optimal' values for $\delta$ as were calculated (as a numerical solution of (\ref{19})) for various choices of $\lambda$, along with the corresponding values of the maximal expected reward ${\cal E}^*_\lambda$, and the matching initial probability of "success" of this $\delta^*-$ Symmetric Strangle strategy. As can be seen from Table 1, the selling of a 'standard' 16-delta symmetric strangle with its $0.68$ 'success' probability should be coupled with an exit strategy  that limits losses  at 100\% of the credit received may yield a maximal expected relative reward of ${\cal E}^*_{1}=0.031$.  In contrast, the selling of a 30-delta symmetric strangle 
of the lesser 'success' probability (of 0.4) should be coupled with an exit strategy  that limits losses  at 25\% of the credit received, but may triple the maximal expected  reward to ${\cal E}^*_{0.25}=0.091$

\section{Appendix A}
In this appendix we study the coordinate-wise behavior of $R(\delta, \nu)$ as given in (\ref{15}) over the practical domain ${\cal B}$. To begin with, note first that since
$$
e^{z_\delta \nu+\nu^2/2}\equiv \frac{\phi(z_\delta)}{\phi(z_\delta+\nu)}\qquad \text{and}\qquad e^{-z_\delta \nu+\nu^2/2}\equiv \frac{\phi(z_\delta)}{\phi(z_\delta-\nu)},
$$
we may express  $R(\delta, \, \nu)$ in (\ref{15}), entirely in terms of the standard normal $pdf$ and $cdf$, as
\be\label{15a}
R(\delta, \, \nu)=  \frac{\phi(z_\delta-\nu)\cdot \Phi(z_\delta+\nu)- \phi(z_\delta+\nu)\cdot \Phi(z_\delta-\nu)}{\phi(z_\delta+\nu)-\phi(z_\delta-\nu)}.
\ee
Upon differentiating expression  (\ref{15a}) of $R(\delta, \nu)$, with respect to $\delta$ and with respect to $\nu$ along with the fact that  $\phi^\prime(u):=\frac{d}{du}\phi(u)=-u\phi(u)$ we obtain the following results. 

\begin{lemma}\label{lm3}{With $R(\delta, \, \nu)$ as defined in (\ref{15}) above we have,
\begin{itemize}
\item[a)] For each fixed $\delta<0.5$ (so that $z_\delta<0$), 
$
\frac{\partial R}{\partial \nu}=\frac{2A}{B^2}(B+z_\delta \cdot D )\leq 0, 
$
\item[b)] For each fixed $\nu>0$,
$
\frac{\partial R}{\partial \delta}=-(z_\delta-\nu)\cdot A\cdot D>0, 
$
\end{itemize}
where, $ D:= (\Phi(z_\delta+\nu)-\Phi(z_\delta-\nu))>0$, $B:= (\phi(z_\delta+\nu)-\phi(z_\delta-\nu))>0$, and  $A:=(\phi(z_\delta+\nu)-\phi(z_\delta-\nu))>0$. 
}
\end{lemma}

\begin{proof} The proofs   of these results, though tedious,  are straightforward to  establish noting that $(z_\delta-\nu)<0$ and   $(B+z_\delta \cdot D)\leq 0$ over ${\cal B}$; the details are omitted. \end{proof}

}

\begin{thebibliography}{999}

\bibitem{} Black F., and Scholes M., (1973).
\newblock The pricing of options and corporate liabilities.
\newblock \emph{The Journal of Political Economy}, 637-654.

\bibitem{} Doran, J., and  Krieger, K., (2010).
\newblock Implications for Asset Returns in the Implied Volatility Skew.
\newblock \emph{Financial Analysts Journal}, 66(1), 65-76.


\bibitem{} Gatheral, J., (2006). 
\newblock \emph{The Volatility Surface,}
\newblock  John Wiley and Sons, NJ.


\bibitem{ } Hull, C. J., (2017).
\newblock \emph{Options, Futures and Other Derivatives},
\newblock 10th Ed, Pearson, NY.


\bibitem{} Iacus, M. S., (2011). 
\newblock \emph{Option Pricing and Estimation of Financial Models with R,}
\newblock   John Wiley and Sons, NY.

\bibitem{} Jiang, L. (2005).
\newblock \emph{Mathematical Modeling and Methods of Option Pricing,}
\newblock Translated from Chinese by Li. C,
\newblock   World Scientific, Singapore.




\bibitem{} Merton, R., (1973).
\newblock Theory of rational option pricing.
\newblock \textit{The Bell Journal of Economics and Management Science}, 141-183.

\bibitem{} Peskir, G. and Shiryaev, A., (2002)
\newblock A Note on the Call-Put Parity and a Call-Put Duality.
\newblock \textit{Theory of Probability and Its Applications }, Vol 46(1), 167-170. 

\bibitem{}  Sinclair, E., (2010). 
\newblock \emph{Option Trading; Pricing and Volatility Strategies and Techniques}
\newblock John Wiley and Sons, NJ


\bibitem{} Wilmott, P., Howison, S. and Dewynne, J., (1995).
\newblock  \textit{The Mathematics of Financial Derivatives: A Student Introduction,}
\newblock   1st Ed,  Cambridge University Press. 

\bibitem{} Yalincak, H. O., (2012).
\newblock Criticism of the Black-Scholes Model: But Why is It Still Used?: (The Answer is Simpler than the Formula).
\newblock Available at SSRN: \url{https://papers.ssrn.com/sol3/papers.cfm?abstract_id=2115141}

\end{thebibliography}
\end{document}